\documentclass{article}

\usepackage{authblk}

\usepackage{graphicx} 

\usepackage[colorlinks=true, allcolors=blue]{hyperref}
\usepackage{amsmath}
\usepackage{graphicx}
\usepackage{enumitem}
\usepackage{xspace}
\usepackage{amsthm}
\usepackage{a4wide}
\usepackage{algorithm2e}
\usepackage[margin=1in]{geometry}

\usepackage{todonotes}

\newtheorem{lem}{Lemma}
\newtheorem{thm}[lem]{Theorem}

\newtheorem{clm}{Claim}[lem]

\newcommand{\MAF}{{\rm MAF}}

\title{Computational support for case-heavy proofs in mathematical phylogenetics\footnote{Corresponding author: steven.kelk@maastrichtuniversity.nl.}}

\author[1]{Luca Forte}
\author[2]{Leo van Iersel}
\author[1]{Steven Kelk}
\author[2]{Ruben Meuwese}

\affil[1]{Department of Advanced Computing Sciences, Maastricht University, The Netherlands.}
\affil[2]{Delft Institute of Applied Mathematics, Delft University of Technology, The Netherlands.}

\date{December 19th 2025}

\begin{document}

\maketitle

\begin{abstract}
In this note we demonstrate that a number of case-heavy combinatorial proofs in the mathematical phylogenetics literature can be proven more compactly using computational support. We use these techniques to also prove several new combinatorial lemmas that would have taken considerable effort to prove by hand. We are optimistic that similar approaches can be deployed more widely in phylogenetics. 
\end{abstract}

\section{Introduction}

Phylogenetics is the discipline of constructing evolutionary (`phylogenetic') trees and networks which summarize how
a contemporary group of species, modelled by the leaves of the tree/network,  evolved from a point of common ancestry \cite{SempleSteel2003}. The mathematical and computational challenge is that we typically only have data available for the species at the leaves, such as DNA data, not for hypothesized ancestors. There are many variants of this inference problem, depending on assumptions that are made about the input data, the underlying evolutionary model and whether trees or networks are desired as output. There are also multiple secondary challenges, such as rigorously quantifying the extent to which two competing candidate trees or networks are truly different. Many of the primary and secondary tasks associated with phylogenetic inference are themselves NP-hard, and even when efficient algorithms exist it can still be highly challenging to characterize situations when the algorithms do and do not perform well. These factors have given rise to a rich vein of research at the interface of computational biology, combinatorics and algorithm design. 

This literature has become increasingly mature, and for many lines of investigation proofs of correctness are reaching the limit of what can plausibly be written, and checked, by hand. This is visible in the number of papers that use some form of heavy case-analysis, see e.g. \cite{chen2015faster,shi2022improved,kelk2024deep,kelk2025bounding,remie,leonie}. Rather than abandoning these `almost saturated' research lines, or waiting for the emergence of some general theoretical framework which will allow these analyses to be subsumed within compact, elegant, hand-written proofs, we here propose a third route. We propose that heavy case-analyses in this field should be replaced as far as possible with \emph{computer-generated and computer-verified case analyses}. 
This note should be regarded as a proof-of-concept for this approach, which fits into the framework of proof-by-computation (of which the 4-colour theorem and the Boolean Pythagorean Triples problem are the most well-known examples \cite{appel1977part1, appel1977part2,HeuleKM16}). We have taken a popular NP-hard problem in the field, that of constructing \emph{maximum agreement forests} (see e.g. \cite{bulteau2019parameterized,ShiEtAl2018,kelk2024deep} for overviews),  where various directions of research are slowing down due to the explosion in cases in the underlying proofs. We give new semi-computational proofs for a number of established theorems, pushing the case complexity into a computational core and leaving a much smaller part to be proven by hand. We then deploy this technique to easily prove a number of new combinatorial statements, which undertaken purely by hand would have take many pages. We then discuss how the technique can be used for exploratory hypothesis generation and semi-automated refutation of false conjectures. We conclude the note by discussing some limitations of our technique and how they might be addressed by widening the use of mechanical enumeration. Due to the generic character of our framework we are optimistic that it will be more broadly applicable in phylogenetics.

We have made a prototype of the computational tools we have developed publicly available at \url{https://github.com/skelk2001/caseautomation} allowing others to posit and test hypotheses using our technique.

\section{Preliminaries}
Our notation closely follows \cite{kelk2024deep}.

Throughout this paper, $X$ denotes a non-empty finite set of \emph{taxa} (informally: species). An {\it unrooted binary phylogenetic tree} $T$ on $X$  is an 
undirected tree whose leaves are bijectively labeled with $X$ and whose other vertices all have degree 3. Note that the single vertex ($|X|=1$) and the single edge ($|X|=2$) are also unrooted binary phylogenetic trees. If $T$ and $T'$ are two unrooted binary phylogenetic trees on
$X$, and there is an isomorphism between $T$ and $T'$ which preserves $X$, then we write $T=T'$. Given that all phylogenetic trees in this note are unrooted and binary, we simply use \emph{phylogenetic tree} as shorthand. 
Two leaves, say $a$ and $b$, of $T$ are called a {\it cherry} $\{a,b\}$ of $T$ if they are adjacent to a common vertex. Moreover, for each $x\in X$, we use $p_x$ to denote the unique neighbor of $x$ in $T$ and refer to $p_x$ as the {\it parent} of $x$.

For $X'  \subseteq X$, we write $T[X']$ to denote the unique, minimal subtree of $T$ that connects all elements in $X'$. For brevity we call $T[X']$ the \emph{embedding} of  $X'$ in $T$. For an edge $e$ of $T$, we say that $T[X']$ {\it uses} $e$, if $e$ is an edge of $T[X']$. Furthermore, we refer to the phylogenetic tree on $X'$ obtained from $T[X']$ by suppressing degree-2 vertices as  the {\it restriction of $T$ to $X'$}  and we denote this by $T|X'$. 

Let $T$ be a phylogenetic tree on $X$.  We say that a subtree of $T$ is {\it pendant} if it can be detached from $T$ by deleting a single edge. For $n\geq 2$, let $K = (\ell_1,\ell_2\ldots,\ell_n)$ be a sequence of distinct taxa in $X$. 
We call $K$ an $n$-chain of $T$ if there exists a path $(p_{\ell_1},p_{\ell_2},\ldots,p_{\ell_n})$ in~$T$ or there exists a path $(p_{\ell_2},\ldots,p_{\ell_n})$ in~$T$ and $p_{\ell_1}=p_{\ell_2}$ or there exists a path $(p_{\ell_1},\ldots,p_{\ell_{n-1}})$ in~$T$ and $p_{\ell_{n-1}}=p_{\ell_n}$.

If  $p_{\ell_1} = p_{\ell_2}$ or $p_{\ell_{n-1}} = p_{\ell_n}$ holds, then $K$ is said to be {\it pendant} in $T$. To ease reading, we sometimes write $K$ to denote the set $\{\ell_1,\ell_2,\ldots,\ell_n\}$. It will always be clear from the context whether $K$ refers to the associated sequence or set of taxa. If a pendant subtree $S$ (resp. an $n$-chain $K$) exists in two phylogenetic trees $T$ and $T'$ on $X$, we say that $S$ (resp. $K$) is a {\it common} subtree (resp. chain) of $T$ and $T'$.

Let $T$ and $T'$ be two phylogenetic trees  on $X$. Furthermore, let $F = \{U_0, U_1,U_2,\ldots,U_k\}$ be a partition of $X$, where each block $U_i$ with $i\in\{0,1,2,\ldots,k\}$ is  referred to as a \emph{component} of $F$. We say that $F$ is an \emph{agreement forest} for $T$ and $T'$ if the following conditions hold. 
\begin{enumerate}
\item [(1)] For each $i\in\{0,1,2,\ldots,k\}$, we have $T|U_i = T'|U_i$. 
\item [(2)] For each pair $i,j\in\{0,1,2,\ldots,k\}$ with $i \neq j$, we have that
$T[U_i]$ and $T[U_j]$ are vertex-disjoint in $T$, and $T'[U_i]$ and $T'[U_j]$ are vertex-disjoint in $T'$. 
\end{enumerate}
\noindent
Let $F=\{U_0,U_1,U_2,\ldots,U_k\}$ be an agreement forest for $T$ and $T'$. The \emph{size} of $F$ is simply its number of components; i.e. $k+1$. Moreover, an agreement forest with the minimum number of components (over all agreement forests for $T$ and $T'$) is called a \emph{maximum agreement forest (MAF)} for $T$ and $T'$. The number of components of a maximum agreement forest for $T$ and $T'$ is denoted by $d_\MAF(T,T')$.

\section{Results}

\subsection{Warming up: long chains are preserved.}
\label{subsec:getstarted}
In this section we introduce many ideas and concepts which will be re-used in later sections.

Let $T$ and $T'$ be two phylogenetic trees on $X$. A well-known result by \cite{AllenSteel2001} proves that (i) collapsing common pendant subtrees to a single taxon does not alter $d_\MAF(T,T')$ and, (ii) common chains of length 4 or more can be truncated to length 3 without altering $d_\MAF(T,T')$. The proof of (1) is elementary. However, the proof of (2) is much more involved and relies on  the correctness of \cite[Lemma 3.3]{AllenSteel2001}. In a nutshell, this proves the following: if $K$ is a common chain and $|K| \geq 3$, then there exists some MAF (i.e. some optimal solution)  $F = \{U_0, U_1, \ldots\}$ such that $K \subseteq U_i$ for some $U_i$. In this case, we say that $F$ \emph{preserves} $K$. The proof of Lemma 3.3 is based on case-analysis.

Later in \cite{kelk2020new} it was pointed out that, although the statement of Lemma 3.3 in \cite{AllenSteel2001} is correct, it misses a number of crucial cases. To remedy this they gave a full proof for a slightly more general statement (Theorem 5), which encompassed several pages of case analysis; the extra length is primarily the result of the completed case analysis rather than the extra generality.

Here we show how to tackle the proof without any significant manual case analysis. In particular, we demonstrate an alternative proof approach based on induction, with the base case generated and checked by computer. In this specific situation the base case can also be checked by hand, but as we see in later examples the base case can quickly become rather large, and then computational support is essential.

This is the exact statement that we will (re-)prove, which in terms of generality lies somewhere between Lemma 3.3 of \cite{AllenSteel2001} and Theorem 5 of \cite{kelk2020new}.

\begin{thm}
\label{thm:3chain}
Let $T$ and $T'$ be two phylogenetic trees on $X$. Let $K$ be a common chain such that (i) $|K| \geq 3$ or (ii) $|K|=2$ and $K$ is pendant in at least one of $T, T'$. Then there exists a MAF that preserves $K$.
\end{thm}
\begin{proof}
If $|K|=2$ then without loss of generality we assume that $K$ is pendant in $T'$.

For the purpose of building a recursive/inductive argument we actually prove a slightly different statement, Claim \ref{clm:core}; this is where most of the work has to be done.  After proving this we will show how to use it to directly prove the theorem.  The notion of a chain (and commonality) extends to forests in the expected way: $K$ is a chain of $F'$ if it is a chain of one of the trees in $F'$. Crucially, in the statement of the claim, the forest $F'$ can be 
\emph{any} disjoint union of unrooted binary phylogenetic trees whose leaves together partition $X$: it is not necessarily an agreement forest. It is also helpful to note that cuttting an edge in a tree of $F'$ has the effect of further refining the partition of $X$ it represents\footnote{After cutting an edge it is common to ``tidy-up'' the obtained forest by repeatedly suppressing any vertices of degree-2. This ensures that each tree in the forest fulfills the definition of being an unrooted, binary phylogenetic tree.}. 
\begin{clm}
\label{clm:core}
Let $T$ and $F'$ be a tree and a forest on the same set of taxa $X$, and let $K$ be a common chain with $|K|=3$ (or a common chain with $|K|=2$ that is pendant in at least one of $T$ and $F'$) that is preserved
in $F'$. If $|K|=2$, then by extending the earlier assumption we can assume without loss of generality that $K$ is pendant in $F'$. Let $k$ be the smallest number of components in any agreement forest of $T$ and $F'$ reached by cutting edges in $F'$. Then there exists an agreement forest with at most $k$ components and where $K$ is preserved, that can be reached by cutting in $F'$.
\end{clm}
\begin{proof}
\renewcommand{\qedsymbol}{\ensuremath{\clubsuit}}
By induction on $|X \setminus K|$. The base case is $|X \setminus K| \leq 2$, we will discuss proving this later. 

So, assume that $|X \setminus K| \geq 3$. Due to this assumption, $T$ contains some cherry $\{p,q\} \subset X$ that does not contain any taxa from $K$. Note that the embeddings $T[\{p,q\}]$ and $T[K]$ are disjoint in $T$.

If $\{p,q\}$ is also a cherry in $F'$, we can collapse the cherry in both $T$ and $F'$ into a single representative taxon $pq$, reducing the number of taxa, and we are done by induction. More formally: in the small $K$-preserving agreement forest on the $X' = (X \setminus \{p,q\}) \cup \{pq\}$ instance, whose existence is guaranteed by the inductive step, we can simply re-expand the $pq$ taxon back into $p,q$ to obtain the desired agreement forest for the $X$ instance, without increasing its size or influencing the preserved status of $K$. Correctness here stems from the well-known fact that collapsing a cherry in $(T,F')$ has no meaningful impact on the space of smallest agreement forests that can be reached from that point.

So $\{p,q\}$ is not a cherry in $F'$. From this we know that in \emph{any} agreement forest $F^{*}$ obtained by cutting in $F'$ at least one of the following situations must happen; these simple, self-contained observations are the foundation of many branching algorithms that have been developed to compute MAF \cite{whidden2013fixed}.
\begin{itemize}
\item[(i)] $p$ is a singleton component in $F^{*}$,
\item[(ii)] $q$ is a singleton component in $F^{*}$,
\item[(iii)] $\{p,q\}$ are both contained in some component of $F^{*}$, in which case all but perhaps one of the edges incident to the $p-q$ path in $F'$, were cut in obtaining $F^{*}$ from $F'$\footnote{Note that if $p$ and $q$ are not together in a component of $F'$, then (iii) cannot hold, so at least one of (i) or (ii) must hold.}. 
\end{itemize}
Consider some agreement forest that can be reached from $F'$ by making $k$ cuts, where $k$ is \emph{minimum} for this specific $T, F'$ combination. We now let $F^{*}$ refer specifically to this minimum reachable forest.

If (i) or (ii) holds for $F^{*}$ then we are done by induction, via the following route. Suppose without loss of generality that (i) holds. 
We delete $p$ from both $T$ and $F'$, obtaining an instance on $X \setminus \{p\}$ to which induction can be applied. It is well-known that, due to $F^{*}$ containing $p$ as a singleton component, a smallest agreement forest for the original $(T,F')$ instance on $X$ can be obtained from an arbitrary smallest agreement forest for the $X \setminus \{p\}$ instance simply by adding $\{p\}$ to the partition. In particular, this applies to the specific agreement forest returned by the inductive step. Adding the taxon $p$ back does not influence the preservation of the chain $K$, so we are done.

If (iii) holds for $F^{*}$, then let $U_{pq}$ be the component of $F^{*}$ that contains both $p$ and $q$.
If $p$ and $q$ are both on the \emph{same side} of $K$ in $F'$ -- i.e. the path from $p$ to $q$ in $F'$ does not intersect with any parents of taxa in $K$ -- we are done because (1) none of the edges incident to the $p-q$ path (some or all of which will be cut) lie on $F'[K]$, ensuring that $K$ is still preserved, and (2) after making those edge cuts there is a common cherry $\{p,q\}$ which can be collapsed into a single taxon, reducing the number of taxa and thus triggering the
induction.

Hence, the only remaining case is the subcase of (iii) when $p$ and $q$ are on \emph{opposite} sides of the chain $K$ in $F'$. (Note that if $|K|=2$ then this case cannot happen, due to $K$
being pendant in $F'$. So at this point the entire proof for $|K|=2$ is already done). We distinguish the following two subcases:
\begin{itemize}
\item[(a)] no taxa from $K$ are in $U_{pq}$. Then the 3 taxa in $K$ are necessarily singletons in $F^{*}$. This is because in $F'$ the path from $p$ to $q$ passes through the parents of the taxa in $K$, but $U_{pq}$ cannot contain any of the taxa from $K$ by assumption. In this situation we can modify $F^{*}$ as follows. First, observe that there exists $x \in \{p,q\}$ such that $F'[ U_{pq} \setminus \{x\}]$ no longer passes through the parents of the taxa in $K$. We cut off exactly one such $x$ into a singleton component. Second, observe that there might be some component $U \in F^{*}$ such that $U \cap K = \emptyset$ and in $T$ (as opposed to in $F'$) $T[U]$ passes through the parents of $K$. Such a component is called a \emph{bypass component}. We apply one cut to $U$ (note that $U$ might be equal to $U_{pq})$ so that in both $T$ and $F'$, no component passes through the parents of $K$. Finally, in $F^{*}$ we replace the three singleton components corresponding to $K$ with a single component exactly equal to $K$. The size of the agreement forest has not increased (we increased the size $F^{*}$ by at most two, and reduced the number of components by at least two), and $K$ is preserved. Moreover, this means that branch (i) or (ii) would have found this agreement forest (due to $x \in \{p,q\}$ being a singleton component), so we can leverage the correctness of those branches.

\item[(b)] $U_{pq}$ contains at least 1 taxon $c$ from $K$. Due to $\{p,q\}$ being a cherry in $T$, it then follows that $U_{pq} = \{p,q,c\}$ and that the other two taxa in $K$, let us call them $a$ and $b$, are both singletons in $F^{*}$. (Note that in $T$ no other component can intersect with, or bypass, $K$, due to $T[U_{pq}]$ passing through the parents of at least one taxon in $K$). We cut $U_{pq}$ twice to obtain singletons $p, q, c$, and then replace the $a,b,c$ singleton components with a single $K$ component. This leads to a new agreement forest that is no larger than $F^{*}$ in which $K$ is preserved (and in which both $p$ and $q$ are singletons). So, again, we can leverage branch (i) or (ii).
\end{itemize}
This concludes the inductive step.

For the \textbf{base case}, we have the following algorithmic procedure.\\

\noindent\fbox{%
\begin{minipage}{0.97\linewidth}
\begin{algorithm}[H]
\For{every tree $T$ on $X$ with $|X \setminus K| \le 2$ that has chain $K$}{
    \For{every tree $T'$ on $X$ that also has chain $K$}{
        \For{every forest $F'$ of $T'$ that preserves $K$}{
            check whether the claim holds for $T, F'$; if it does not, return FALSE. Otherwise continue.
        }
    }
}
return TRUE.
\end{algorithm}
\end{minipage}%
}
\vspace{1em}\\
There are many different ways to perform the $(T,F')$ check at the centre of the loop. In our implementation we use the Integer Linear Programming (ILP) formulation of Van Wersch et al \cite{van2022reflections}. This is a natural ILP for computing a MAF of two trees $T, T'$ in which there is a binary decision variable for each edge of $T'$. Setting apriori a variable to 0 means ``definitely do not cut this edge'' and to 1 means ``definitely cut this edge''. To model $F'$ we simply set edges to 1 that can be cut to obtain $F'$ from $T'$. This ensures that the ILP can correctly compute the minimum number of cuts to obtain an agreement forest from $(T,F')$. 
We run the ILP twice. The first time, we simply solve the ILP to compute this minimum number of cuts OPT1. Next, we add constraints setting all the decision variables corresponding to edges on $F'[K]$ to zero: this models preservation of $K$. We solve the modified ILP to obtain a new minimum OPT2. The claim (i.e. that $K$ is preserved by some MAF) holds if and only if OPT1 is the same as OPT2. The generation and solving of the base case is shown in Appendix \ref{subsec:bc3chain}.
\end{proof}
Now that Claim \ref{clm:core} has been established, taking $F' = T'$ immediately establishes the theorem for the situation when $|K| \leq 3$.
To establish the theorem for $|K|>3$, we apply the following argument standard in the literature; we repeat it here simply for completeness. We truncate $K$ to a length 3 chain $K' = (a,b,c)$. Invoking the claim for $|K|=3$ shows that there exists a MAF that preserves $K'$. The taxa in $K \setminus K'$ can now be inserted back into $T$ and $T'$  without increasing the size of the agreement forest: we simply add the taxa $K \setminus K'$ to the component that contains $\{a,b,c\}$.
Combined with the fact that the size of a MAF is non-increasing under the deletion of taxa, we obtain a MAF in which all of $K$ is preserved. 
\end{proof}

\noindent
\emph{Remark.} When solving the base case, it is safe to ignore $(T,F')$ pairs where $T$ contains a cherry (outside
$K$). That's because the inductive argument can also be applied here, reducing the number of taxa by at least one. In other words: the check for such $(T,F')$ reduces to a check on a smaller instance within the base case\footnote{Note, however, that such pruning inside the space of base case instances is only safe
if this space is defined via a \emph{correct} inductive step. When hunting for counter-examples and a correct inductive step is not known, this pruning should be avoided as it can cause counter-examples to be missed. We return to this issue in Section \ref{sec:exploration}. }. Applying a similar
logic, $(T,F')$ pairs can also be disregarded where $F'$ contains at least one singleton component.

\subsection{Interrupted 4-chains}
\label{subsec:int4}

Although the machinery introduced in the previous section seems rather heavy relative to the statement being proven, the same machinery can now be deployed with only small modifications to prove more ambitious statements.

Consider the situation depicted in Figure \ref{fig:interrupted4chain}. Here we say that $K = (a,b,c,d)$ is an \emph{interrupted 4-chain}. Informally: $(a,b,c,d)$ would be a common 4-chain of $T$ and $T'$, except for the presence of the dotted edge
in the figure (that intersects the path from $b$ to $c$ in $T$), which we call the \emph{interrupter}\footnote{For the proof that follows it is convenient to also allow $B=\emptyset$. In this case the interrupted 4-chain is simply a common 4 chain, and the interrupter edge doesn't actually exist.}. The most powerful set of data-reduction rules for the MAF problem requires a proof that if this interrupted 4-chain structure is present, then at least one MAF has the property that none of its components use the interrupter edge. The existing proof of this claim uses an extensive hand-written case analysis \cite{kelk2024deep}. We can replace this hand-written analysis with a more compact proof that makes
greater use of computation.

\begin{figure}[h]
\centerline{\includegraphics{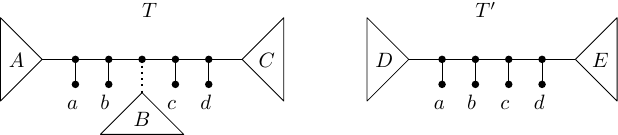}}
\caption{\label{fig:interrupted4chain}An interrupted 4-chain on taxa~$a,b,c,d$.
}
\end{figure}

\begin{thm}[Theorem 3.2, \cite{kelk2024deep}] \label{thm:interrupted4chain}
    Let $T$ and $T'$ be two phylogenetic trees on $X$ and let $K=(a,b,c,d)$ be an interrupted 4-chain of $T$ and $T'$, as shown in Figure \ref{fig:interrupted4chain}. Then there exists a maximum agreement forest $F$ for $T$ and $T'$ such that, for each $U\in F$, $T[U]$ does not use the interrupter of $K$ in $T$ (equivalently, $U$ does not contain taxa from both $B$ and $X\setminus B$).
\end{thm}
\begin{proof}
We use a very similar induction argument to that in Theorem \ref{thm:3chain}. Rather than directly proving that some MAF does not use the interrupter edge, we will prove a very slightly different result which implies what we need. Namely: that some MAF has \emph{all the taxa $a,b,c,d$ together in one component} (possibly with other taxa). Once a component contains all four of these taxa, it cannot contain any taxa from $B$, because then the component would have a different topology in $T$ and $T'$; hence, the interrupter is not used.

The exact claim we will prove is as follows; correctness will then follow by taking $F'=T'$ in this claim. 

\begin{clm}
\label{clm:int4chain}
Let $T$ and $F'$ be a tree and a forest on the same set of taxa $X$, and let $K$ be an interrupted 4-chain whereby $K$ is preserved in $F'$ i.e. some component of $F'$
completely contains the chain $K$. Let $k$ be the smallest number of components in any agreement forest of $T$ and $F'$ reached by cutting edges in $F'$. Then there exists an agreement forest with at most $k$ components and where $K$ is preserved, that can be reached by cutting in $F'$.
\end{clm}

 First, we assume that $| X \setminus K| \geq 4$ (rather than $\geq 3$). Under this assumption then by the pigeonhole principle at least one of the sets $A, B, C$ in $T$ must contain a cherry $\{p,q\}$. Next, we observe that cases (i), (ii) go through unchanged, as does case (iii) when $p$ and $q$ are on the same side of $K$ in $F'$. Cases (iii)(a) and (iii)(b) require some minor modifications, however. We start with (iii)(b). As before, $U_{pq}$ can intersect with at most one of the taxa $x \in \{a,b,c,d\}$, and in that case $U_{pq} = \{p,q,x\}$ and all the taxa in $\{a,b,c,d\} \setminus \{x\}$ will be singletons. We start by cutting $U_{pq}$ twice to obtain the singleton components $p, q, x$. 
As before the goal is to introduce $\{a,b,c,d\}$ as a single component, which will compensate for these extra cuts, ensuring that we still have a smallest agreement forest, with the desirable property that at least one of $p,q$ is a singleton component (which lets us use branch (i) or (ii)). However, before we can introduce  $\{a,b,c,d\}$ as a single component, there is still potentially some further cutting to be done. Specifically, there might be a component $U \neq U_{pq}$ where $U \cap K = \emptyset$ such that $U$ intersects with $B$ and exactly one of the sets $A, B, C$. (Other patterns of intersection with $A, B, C$ are excluded because $T[U_{pq}]$ already passes through at least one parent of a taxon in $K$). This has the effect that in $T$, $T[U]$ enters the $K$ region via $B$ and leaves via $A$ or $C$; this could happen, for example, if $\{p,q,x\} \subseteq (A \cup \{a,b\})$. To deal with this we apply a single cut to $U$ so that its image no longer intersects with $T[K]$. Hence, in total we make at most 3 cuts. But fortunately this is compensated for by merging $a,b,c,d$ into a single component, saving 3 cuts, and thus ensuring we still have a smallest agreement forest in which $K$ is preserved.

For case (iii)(a), we start by again cutting off one of $p$ and $q$, using the same considerations as given earlier when choosing which of the two to cut. We might however have to apply 1 or 2 further cuts to `free up' the $a,b,c,d$ region in $T$. This could happen, for example, if some component (possibly $U_{pq}$ itself) has intersection with two or three of $A, B, C$ (but no intersection with $a,b,c,d$); 2 cuts will be required if $U$ intersects with all three of $A,B,C$. Fortunately these in total 3 cuts are compensated for when we merge the 4 singletons $\{a,b,c,d\}$ into a single component. So once again we have a smallest agreement forest in which at least one of $p,q$ is a singleton component, and $K$ is preserved, allowing us to leverage branch (i) or (ii).

All that remains is to check the base case, which now consists of all $(T, F')$ pairs where $|X \setminus K| \leq 3$. 
The proof of the base case is once again computational, and summarized in Appendix \ref{subsec:bcinterrupt}.
\end{proof}

\subsection{Moving on: new combinatorial lemmas}
\label{subsec:newlemmas}
In this section we demonstrate the potential of our method by proving two new (related) combinatorial statements that would take many pages of hand-written case-analysis if proven the traditional way. We say that two trees $T, T'$ have a \emph{common 2-2-2 star on} $\{a,b,c,d,e,f\}$ if they have the structure shown in in Figure \ref{fig:3star}. 

\begin{figure}[h]
\centerline{\includegraphics{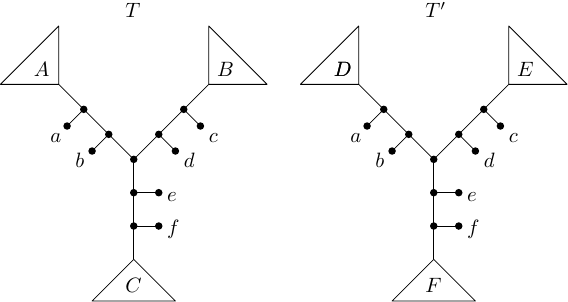}}
\caption{\label{fig:3star}Two trees with a common 2-2-2 star on $\{a,b,c,d,e,f\}$. Note that it is not required for any of the sets $A,\ldots ,F$ to be non-empty.
We will prove that there exists a MAF of $T, T'$ such that some component of the MAF contains all of $a,b,c,d,e,f$ i.e. that the common 2-2-2 star is preserved by some optimal solution.}
\end{figure}

\begin{thm}
\label{thm:32} Suppose $T$ and $T'$ have the common 2-2-2 star structure shown in Figure \ref{fig:3star}. Then there exists a MAF in which the common 2-2-2 star is preserved.
\end{thm}
 
 \begin{proof}
 For the inductive step $|X \setminus K| \geq 4$ is sufficient (where here $K=\{a,b,c,d,e,f\}$), because this ensures that at least one of $A,B,C$ has a cherry $\{p,q\}$. So the base-case is $|X \setminus K| \leq 3$; its computational proof is summarized in Appendix \ref{subsec:bc32}.

The inductive step has a similar structure to earlier arguments. We know that $p,q$ are both in exactly one of $A, B, C$, although for our analysis it does not matter which one. If $p$ and $q$ are both entirely inside $D$, $E$ or $F$ we are done.
So, as usual, the main case to consider is when $p \in P$ and $q \in Q$ where $P, Q$ are two distinct sets from $\{D,E,F\}$.

The path from $p$ to $q$ in $F'$ passes through the parents of exactly 4 of the taxa in $K$; let us denote this subset $K_4$. Note that at least three of $K_4$ are necessarily singletons (because $p,q$ is a cherry in $T$). 

We proceed as follows. Define $U_{pq}$ as usual. We cut \emph{both} $p$ and $q$ off, requiring at most 2 cuts, and causing $p$ and $q$ to become singletons. (If $U_{pq}$ also contained exactly one taxon from $K_4$, then this taxon also becomes a singleton, by virtue of the fact that the component necessarily had size exactly 3). Now, suppose without loss of generality that $K_4 = \{a,b,e,f\}$. We make at most two further cuts: one cut on the `edge' between $f$ and $C$\footnote{More formally, this is the edge of the component whose embedding passes through the edge between $f$ and $C$ in $T$.} to cut a component that has taxa in $C$ and whose embedding in $T$ intersects with the parent of $f$ in $T$; and by the same logic a second cut on the `edge' between $A$ and $a$.
Note that 2 cuts are only required if $U_{pq} \cap K = \emptyset$ and some component intersects simultaneously with $A, B$ and $C$. This has the effect of `freeing' the region occupied by $K_4$ in both $T$ and $F'$ at the total expense of 4 cuts. Now, the four singletons $a,b,e,f$ can be replaced by a single component $\{a,b,e,f\}$, saving 3 cuts. Finally, we merge the new $\{a,b,e,f\}$ component with all components that intersect $\{c,d\}$; there will be at least one. This will save a 4th cut, and we again have a smallest agreement forest in which the 2-2-2 star is preserved and at least (in fact, both) of $p,q$ are singletons, so branch (i) and (ii) apply.
\end{proof}
In fact, we can prove something stronger. A common 2-2-1 star on $\{a,b,c,d,e\}$ is defined similarly to a 2-2-2 star (see Figure \ref{fig:3star}) except that the taxon $f$ is missing. 

\begin{thm}
\label{thm:221}
Suppose $T$ and $T'$ have a common 2-2-1 star on $\{a,b,c,d,e\}$. Then there exists a MAF in which the 2-2-1 star is preserved.
\end{thm}
\begin{proof}
The proof follows the same structure as the proof of Theorem \ref{thm:32}. Indeed, large parts of the proof go through unchanged, as we now show. If the path from $p$ to $q$ in $F'$ passes through the parents of 4 of the taxa in $K = \{a,b,c,d,e\}$, which in this case are necessarily $\{a,b,c,d\}$, no non-trivial changes are required to the proof at all. 
The crucial fact is that after freeing up the parents of $\{a,b,c,d\}$ with at most 4 cuts, and merging the singletons $a,b,c,d$ to save 3 cuts, the component intersecting $e$ can be merged with $a,b,c,d$ to save the final cut (and this saving occurs even though $f$ is absent). 

So we only need to consider the case when the path from $p$ to $q$ passes through the parents of exactly 3 taxa in $K$, say $K_3=\{a,b,e\}$.  If $U_{pq}$ contains a taxon from $K_3$ then the parents above $K_3$ can be freed up in both $T$ and $F'$ with at most 3 cuts: cutting $p$ off, cutting $q$ off, and cutting at most one component (distinct from $U_{pq}$) that passes through the $K_3$ region in $T$. Merging the three singletons corresponding to taxa in $K_3$ saves 2 cuts. The final cut can be saved by once again merging with the at least one component that intersects with $\{c,d\}$.

Hence, the remaining case is when $U_{pq}$ does not contain any taxa from $K_3$. The problem here is that if 4 cuts are made to free up the $K_3$ region in both $T$ and $F'$, only 2 cuts are saved when merging the 3 singletons $a$, $b$, $e$ into a single component. This is not a problem if (after freeing up the $K_3$ region) there are \emph{two} or more components that intersect with $\{c,d\}$, since combining these components with the merged $K_3$ then saves at least 2 cuts. We call this the \emph{positive situation}. If $U_{pq}$ intersects with exactly one of the two taxa $\{c,d\}$, then the positive situation definitely occurs, so we are done. Similarly, if $U_{pq}$ passes through the parents of $c$ and $d$ in $F'$ but does not contain either $c$ or $d$, then $c$ and $d$ are both singletons so the positive situation definitely occurs. On the other hand, if $U_{pq}$ does not pass through the parent of $c$ in $F'$, then only 3 cuts rather than 4 were required to free up the $K_3$ region in the first place (one cut for $U_{pq}$, freeing up the $K_3$ region in $F'$, and possibly 2 cuts to free up the $K_3$ region in $T$): so the original saving of only 1 cut, rather than 2, is again enough. This brings us to the very final case: $U_{pq}$ contains both $c$ and $d$. Here we make 2 cuts to $U_{pq}$ to create components $\{p\}, \{q\}$ and $U_{pq} \setminus \{p,q\}$ (where necessarily $c$ and $d$ are both contained in the last component).

If the $\{p,q\}$ cherry is in $A$ or $C$, then no further cuts are required to free up the $K_3$ region in $T$, and the counting goes through easily. But what if the cherry is in $B$? Crucially, $U_{pq}$ cannot then contain any taxon $x \in A \cup C$. If it did, $\{p,q,c,d,x\}$ would have a different topology in $T$ than in $F'$ -- informally, $c$ and $d$ would occur in reversed order in $T|U_{pq}$ compared to $F'|U_{pq}$ -- contradicting the definition of an agreement forest. Hence, all that is required to free up the $K_3$ region in $T$ is to cut possibly one component (distinct from $U_{pq}$) whose embedding passes from $A$ to $C$. So 3 cuts are required to free up the $K_3$ region in total, which is sufficient (because we have avoided the bottleneck of requiring 4 cuts) and we are done.

To conclude the induction proof it suffices to prove the claim for $|X \setminus K| \leq 3$. We do this computationally, and this is summarized in Appendix \ref{subsec:bc221}.
\end{proof}

Although the proof of Theorem  \ref{thm:221} is significantly shorter than if it had been proven in the classical way, the inductive step still involves a number of cases. Also, quite a large amount of domain expertise is required to follow the quite intricate  agreement forest arguments. This raises the question of whether
the case analysis \emph{within the inductive step itself} can be automated. Closely relatedly: in each case in the inductive step we apply a cut-and-merge strategy to trigger the inductive claim. Can the search for and use of such strategies also be automated? We return to all these questions later in Section \ref{sec:luca}. Before that we turn to some interesting extra uses of our existing inductive framework.

\section{Exploratory hypothesis testing, disproof via base cases and from ``there exists'' to ``for all'.}
\label{sec:exploration}

An interesting ``inverted'' use of our methodology is follows. Suppose we wish to determine whether a new statement $S$ is true or not. If we can build an inductive argument -- in the same way as the earlier examples in this note -- and the inductive argument itself is correct, then the statement is true \emph{if and only if} the base case holds. This means that if the base case breaks, it gives us a concrete example of input trees for which the claim is false. This gives us a powerful tool for exploratory testing of hypotheses -- and if they are false, understanding why.

There is a subtlety here, however. All the proofs in this note are ultimately based on claims for $(T,F')$ pairs, not $(T,T')$ pairs: this generality is needed to make the inductive argument go through. Could it happen that a $(T,F')$ instance in the base case is false, but that the original statement (on two trees) is true? Fortunately, this is not the case. As the following lemma shows, it is always possible to convert an $(T,F')$ instance for which the claim is false, into a slightly larger pair of trees $(T_\alpha, T_\beta)$ for which the claim is also false.

\begin{figure}[h]
\centerline{\includegraphics{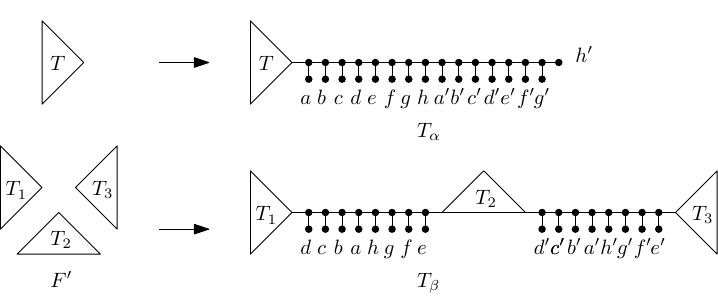}}
\caption{\label{fig:gadgetblock}A $(T,F')$ pair can be modified to obtain two trees $(T_\alpha, T_\beta)$ such that smallest agreement forests reachable from $(T,F')$ and $(T_\alpha,T_\beta)$ are in bijection. In $T_\alpha$ the blocking gadgets are chains $(a,b,c,d,e,f,g,h)$ and in $T_\beta$ they are chains $(d,c,b,a,h,g,f,e)$. From Lemma \ref{lem:block} each blocking gadget is split into two isolated components in 
every maximum agreement forest of $(T_\alpha, T_\beta)$.
}
\end{figure}

\begin{lem}
\label{lem:connect}
Let $(T,F')$ be a pair for which a particular combinatorial statement $S$ of the form ``amongst the smallest agreement forests reachable from this point, there exists at least one with desirable property $S$'' does \underline{not} hold. Then there exists a pair of trees $(T_\alpha, T_\beta)$ for which the statement also does not hold.
\end{lem}
\begin{proof}
Let $k$ be the number of components in $F'$. We introduce $(k-1)$ new chains, each on 8 taxa, which are inserted between the components of $F'$ to connect the components together into a tree $T_\beta$. In $T_\alpha$ the order of taxa in each chain is $a,b,c,d,e,f,g,h$. In $T_\beta$ the order of the taxa in the chain is $d,c,b,a,h,g,f,e$. In $T_\alpha$ the chains can be placed in arbitrary locations outside $T$. These chains are called \emph{blocking gadgets} (based on an idea in \cite{kelk2024agreement}). See Figure \ref{fig:gadgetblock}.  Lemma \ref{lem:block} in the appendix proves (via Lemma \ref{lem:4everywhere} presented later in this section) that, for every blocking gadget $\{a,b,c,d,e,f,g,h\}$, \emph{every} MAF of $(T_\alpha, T_\beta)$ contains components $\{a,b,c,d\}$ and $\{e,f,g,h\}$. In other words, every MAF of $(T_\alpha, T_\beta)$ isolates all the blocking gadgets into their own components, which means -- ignoring the blocking gadgets -- that this MAF can be obtained by cutting in $F'$.
This means that every smallest agreement forest reachable from $(T,F')$ (that has property $S$) can be directly mapped to a MAF of $(T_\alpha, T_\beta)$ (that has property $S$), and vice-versa.
\end{proof}

Even if the inductive step has not been written down, it is still possible to search for counter-examples within the (purported) base case, or slightly beyond it. We give here a number of illustrative examples. 

\begin{enumerate}
\item 
\begin{figure}[h]
\centerline{\includegraphics{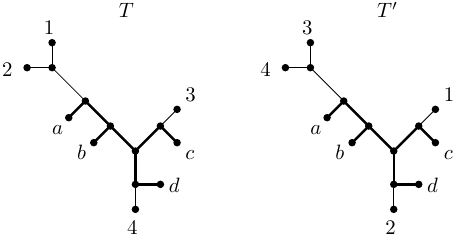}}
\medskip
\centerline{\includegraphics{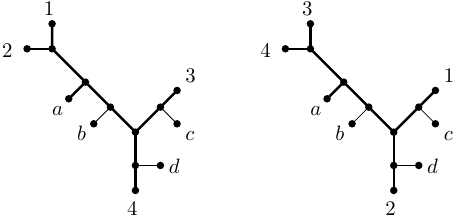}}
\caption{\label{fig:211star} The common 2-1-1 star structure at the centre of the figure (the part spanning the taxa $a,b,c,d$) is not guaranteed to be preserved by some MAF. In the two trees shown, preserving this region leads to strictly sub-optimal solutions with at least~$5$ components (indicated in bold in the top copy of the trees). An optimal solution has only~$4$ components (indicated in bold in the bottom copy of the trees).} 
\end{figure}

Consider Figure \ref{fig:211star}. Is this common 2-1-1 star structure preserved by some MAF? It is not immediately obvious. Rather than attempting to prove it by hand, we ran our base case solver for $|X \setminus K| \leq 4$\footnote{We used the command \texttt{./disproofExistsOnlyTrees.sh ./211star.input 4}. The counter-example was the 1843th tree pair considered, found after 97 seconds. Solve times for speculative counter-example detection tend to be higher than for positive proofs because pruning inside the ``base case'' is not allowed and because counter-examples often exist at higher values of $|X \setminus K|$.}. For $|X \setminus K| \in \{0,1,2,3\}$ no counter-examples were found. However,
at $|X \setminus K| = 4$ we found the counter example, also shown in Figure \ref{fig:211star}. 
A MAF of these two trees contains 4 components (e.g. $\{3, 4, a, 1, 2\}, \{b\}, \{c\}, \{d\}$) but if the 2-1-1 region is preserved, at least 5 components are required. Presumably we would have eventually found this counter-example by hand, but the use of computers greatly accelerated the process.

\item Are common 2-chains always preserved by some MAF? The answer is \emph{no}. Our code found a counter-example at $|X \setminus K|=5$ and this is the \emph{smallest} such counter-example\footnote{We used the command
\texttt{./disproofExistsOnlyTrees.sh ./2chain.input 5}. The counter-example was the 1907th tree pair considered, found after 100 seconds.}. The counter example consists of two linear trees (so-called caterpillars). In tree $T$ the taxa are in the order $1, a, b, 2, 3, 4, 5$ and in $T'$ the order $4, a,b, 5, 3, 1, 2$. If the 2-chain is preserved, at least 4 components are needed: for example $\{a,b,2,3\}, \{1\}, \{4\},$ and $\{5\}$. However, if it is not necessary to preserve the chain, 3 components are sufficient: $\{2,1,3,4,5\}, \{a\}, \{b\}$.

\end{enumerate}

Our framework can be quite easily extended to testing statements about \emph{all} MAFs, not just one. Recall Theorem \ref{thm:interrupted4chain}, about the interrupted 4-chain. We proved this statement by establishing that \emph{some} MAF preserves the 4-chain $K$, and hence the interrupter edge cannot be used by any component of that MAF. Could it be true that every MAF preserves $K$? The answer is \emph{no}: a counter-example exists already at $|X \setminus K|=2$. We found this quickly by running our base-case checking code and testing ``is $K$ preserved in \emph{every} smallest MAF reachable from $(T,F')$?'' The ILP can easily be adapted for this purpose: we first apply the ILP to $(T,F')$ to obtain OPT1. Next we add constraints which enforce that \emph{at least one} edge of the chain $K$ is cut, obtaining OPT2. If OPT1 and OPT2 are equal, the claim is false for this $(T,F')$ (because preservation is not required to attain optimality), otherwise the claim is true. The counter example is: $p \in B$, $q \in C$, $A=\emptyset$;  $p \in E$, $q \in D$. If $K$ is preserved, 3 components are necessary, namely $\{a,b,c,d\}$, $\{p\}$, $\{q\}$. But this can also be achieved by $\{p,c,d\}, \{a,b\}, \{q\}$, which breaks up the chain and uses the interrupter edge.

This adaptation from some MAF to all MAFs can also be used to fairly easily prove stronger statements (when they are true). For example:

\begin{lem}
\label{lem:4everywhere}
Let $K$ be a 
common $4$-chain. Then \underline{every} optimal solution to MAF preserves $K$.
\end{lem}
\begin{proof}
If we trace through the proof of
Theorem \ref{thm:3chain}, observe that a strengthened inductive argument goes through: in fact, (iii)(a) and (iii)(b) can no longer happen, as this would contradict the assumption that $F^{*}$ was the \emph{smallest} possible agreement forest reachable from the given starting point (due to more cuts being saved by merging, than are spent in the cutting phase). So it remains only to prove 
for every $(T,F')$ (where $|X \setminus K| \leq 2$) that $K$ is preserved in \emph{every} optimal solution. This can be verified by hand or by computer, using the modification to the ILP discussed above (see Appendix \ref{subsec:bc4all}).
\end{proof}

To see that Lemma~\ref{lem:4everywhere} can not be strengthened to 
$3$-chains, consider Figure~\ref{fig:3chain}.
\begin{figure}
    \centering
    \includegraphics{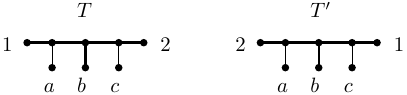}
    \caption{Example of two trees~$T,T'$ with a common $3$-chain and an optimal solution to MAF (indicated in bold) that does not preserve the $3$-chain.}
    \label{fig:3chain}
\end{figure}

\section{Next steps: explicit enumeration both inside and outside induction proofs.}
\label{sec:luca}

As we have shown our inductive framework helps to simplify case-heavy hand-written proofs in this area. There are still areas for improvement however. We list these and then propose a unified response to be explored in future work.

\begin{enumerate}
\item The original Theorem 5 of \cite{kelk2020new} is not just about preserving a single chain but actually about simultaneously preserving a taxa-disjoint \emph{set} of chains. We observed that our proof for single chains (Theorem \ref{thm:3chain}) does not easily extend to sets of chains. The problem is that every cherry $\{p,q\}$ of $T$  might intersect with at least one of the chains in the set. It is possible to deal with this but it involves further technical analysis. Interestingly, the case-heavy hand-written proof of Theorem 5 of \cite{kelk2020new} starts with the argument that a single chain is preserved, and then extends it with very little effort to multiple chains. This is because the hand-written proof, which has a similar cut-and-merge core to our proof, is emphatically \emph{local}. The cutting and merging only affects one chain at a time \emph{and does not impact upon other chains}. Hence, the chains can be preserved iteratively, one at a time. Our inductive attack does not have this explicitly local character, which is why extending it to multiple chains is non-trivial.
\item As the proof of Theorem \ref{thm:221} demonstrates the inductive step can itself remain quite case-heavy, even if a chunk of the complexity has been pushed into a base case. If this case analysis becomes too heavy, it nullifies the advantage of the computational base case. 
\item Once the cases within the inductive step have been identified, each case needs to be dealt with in order to leverage the inductive assumption (which is assumed to hold for smaller numbers of taxa). In this note we have applied ad-hoc cut-and-merge arguments to show that a smallest agreement forest reachable from this point satisfies the combinatorial statement in question. However, the identification and use of these cut-and-merge strategies still requires quite some domain expertise.
\end{enumerate}
Issues 1-3 are very closely related. We believe that the solution lies at least partly in an earlier unpublished version of our work in which we did not use an inductive framework at all. Instead, we tried to mechanically enumerate all the different ways the images of MAF components can interact with the chain region in both trees. The key observation here was that the exact topologies of the parts of the trees outside the chain region only have limited significance for these interactions -- and thus this external region can be represented very compactly with a small set of `meta'-taxa\footnote{This is similar in spirit to the way dynamic programs over tree decompositions leverage small separators to obtain algorithmic efficiency.}. For each mechanically enumerated interaction (addressing issue 2), we can inspect whether it obeys the combinatorial statement, and if it does not, show that some local (addressing issue 1) cut-and-merge strategy is sufficient to restore the combinatorial statement without losing optimality (addressing issue 3). Intriguingly, if the `ground rules' for a legal cut-and-merge strategy can be formalized (requiring e.g. the formalization of agreement forests in terms of quartets), almost the \emph{entire pipeline} could be automated, since all interactions and possible cut-and-merge strategies exist and can thus be found in a finite, enumerable space. Partial deployments are also possible. For example, the enumeration of interactions could be used inside our inductive framework to ensure that all the cases in the inductive step (see issue 2) have been covered.

Such an approach potentially brings more flexibility than our current, primarily inductive framework, and lowers the dependency on domain expertise. The downside is that it requires considerable groundwork to set up the mathematical foundations for the framework, and the running times of the computational parts might be significantly larger (due to having larger spaces to enumerate). We therefore defer this work to a follow-up publication.

\section{Conclusion}

In this note we have shown how some of the complexity in hand-written proofs in phylogenetics can be pushed into a computational core. We have focussed here on static, combinatorial statements about agreement forests, but we are confident that the relevance reaches further. These combinatorial statements (``some optimal agreement forest has property $S$'') are often the foundation of the correctness of reduction rules (kernelization) and branching algorithms. Moreover, similar frameworks can also plausibly be deployed to help tighten bounds on the size of kernels, and to discover new branching rules. We are also confident that similar techniques will work for many non-agreement forest models in phylogenetics, since the canonical ``show by local modification that there is at least one optimal solution with a certain desirable property'' is highly generic.
Alongside the extensions discussed in Section \ref{sec:luca} and the ever-present need to keep running times reasonable, we anticipate very many opportunities for future research.

\section{Acknowledgements}
Ruben Meuwese was supported by grant~OCENW.GROOT.2019.015 from the Dutch Research Council (NWO).

\bibliography{main}{}
\bibliographystyle{plain}

\section{Appendix}
Running times, where reported, are on a Samsung Galaxy Book4 Edge with a Snapdragon(R) X Elite - X1E80100 - Qualcomm(R) Oryon(TM) CPU (3.42 GHz) and 16G of RAM.

\subsection{Computational proof of base-case for preservation of common 3-chain.}
\label{subsec:bc3chain}

There are actually two executions here: one for the case when $K$ is a common 3-chain, and the second is for when $K$ is a common 2-chain that is pendant in $T'$. 

\subsubsection{$K$ is a 3-chain}
The input to our code is:
\begin{verbatim}
(a,(1,(2,(3,b))));
(a,(1,(2,(3,b))));
\end{verbatim}
The code assumes that the part of each tree spanned by the numerical taxa (1,2,3) is the preserved region. Other taxa (here $a$ and $b$) are simply placeholders where the taxa from $X \setminus K$ will be added. There is no correspondence between the $a$ and $b$ in the first tree and the same labels in the second tree.

Here the base case concerns $|X \setminus K| \leq 2$. We solve the base case as follows:
\begin{verbatim}
./proofExistsPrune.sh 3chain.input 2
\end{verbatim}
Here 3chain.input simply contains the two trees presented above. Note that this variant of the code discards any pairs $(T, F')$ where $T$ contains at least one cherry outside the preserved region, or $F'$ contains at least one singleton component; this is what \emph{prune} refers to in the command name. This is safe, and reduces the number of cases somewhat (see the discussion at the end of Section \ref{subsec:getstarted}). On the GitHub page an execution trace is also provided where there is no such pruning. The execution trace is below. Note that the taxa ``$X_i$'' are the taxa from $X \setminus K$. The ``$x_i$'' used to describe each of the forests $F'$ have a different meaning (they correspond to decision variables in the underlying ILP, which are in bijection with edges of $T'$).

\begin{verbatim}
Reading trees from 3chain.input.
Will add up to 2 taxa.
REMOVECHERRIES = yes
LOGS = no
STARTTAXA = 0
FORALL = no
ALLOWFORESTS = yes
ALLOWSINGLETONFORESTS = no
Note: in the execution below, FOREST1 indicates the situation F' = T'.
*****************************************************************
First tree: (a,(1,(2,(3,b))));
Second tree: (a,(1,(2,(3,b))));
*****************************************************************
Going to try expansions with 0 taxa, ... all the way up to 2 taxa.
*****************************************************************
********* OUTERLOOP: |X \setminus K| has exactly 0 taxa. ********
*****************************************************************
Found 1 trees before symmetry reduction:
(1,(2,3));
Found these unique patterns:
(1,(2,3));
1 patterns in total.
Selecting one tree per pattern.
(1,(2,3));   (1,(2,3));
1 canonical trees generated.
Putting these canonical trees (without their pattern) in file ./firstTreesPreCherry.
(1,(2,3));
=======================
Expansions first tree:
(1,(2,3));
Removing expansions of first tree that have a cherry in the first tree
(outside the preserved region).
0 of the trees had such a cherry.
Expansions of first tree after removing trees with cherries:
(1,(2,3));
=======================
Expansions second tree:
(1,(2,3));
=======================
Summary:
1 first trees
1 second trees
Total number of tree pairs:
1
=======================
======================
Combination 1 1 (i.e. tree pair 1)
(1,(2,3));
(1,(2,3));
Considering all the forests. (Singleton forests allowed = no)
FOREST1: x1 <= 1
OPT without modification: 0
OPT with modification: 0
*****************************************************************
********* OUTERLOOP: |X \setminus K| has exactly 1 taxa. ********
*****************************************************************
Found 2 trees before symmetry reduction:
(1,(2,(3,X1)));
(X1,(1,(2,3)));
Found these unique patterns:
(1,(2,(3,Y)));
(Y,(1,(2,3)));
2 patterns in total.
Selecting one tree per pattern.
(1,(2,(3,X1)));   (1,(2,(3,Y)));
(X1,(1,(2,3)));   (Y,(1,(2,3)));
2 canonical trees generated.
Putting these canonical trees (without their pattern) in file ./firstTreesPreCherry.
(1,(2,(3,X1)));
(X1,(1,(2,3)));
=======================
Expansions first tree:
(1,(2,(3,X1)));
(X1,(1,(2,3)));
Removing expansions of first tree that have a cherry in the first tree
(outside the preserved region).
0 of the trees had such a cherry.
Expansions of first tree after removing trees with cherries:
(1,(2,(3,X1)));
(X1,(1,(2,3)));
=======================
Expansions second tree:
(1,(2,(3,X1)));
(X1,(1,(2,3)));
=======================
Summary:
2 first trees
2 second trees
Total number of tree pairs:
4
=======================
======================
Combination 1 1 (i.e. tree pair 1)
(1,(2,(3,X1)));
(1,(2,(3,X1)));
Considering all the forests. (Singleton forests allowed = no)
FOREST1: x1 <= 1
OPT without modification: 0
OPT with modification: 0
======================
Combination 1 2 (i.e. tree pair 2)
(1,(2,(3,X1)));
(X1,(1,(2,3)));
Considering all the forests. (Singleton forests allowed = no)
FOREST1: x1 <= 1
OPT without modification: 1
OPT with modification: 1
======================
Combination 2 1 (i.e. tree pair 3)
(X1,(1,(2,3)));
(1,(2,(3,X1)));
Considering all the forests. (Singleton forests allowed = no)
FOREST1: x1 <= 1
OPT without modification: 1
OPT with modification: 1
======================
Combination 2 2 (i.e. tree pair 4)
(X1,(1,(2,3)));
(X1,(1,(2,3)));
Considering all the forests. (Singleton forests allowed = no)
FOREST1: x1 <= 1
OPT without modification: 0
OPT with modification: 0
*****************************************************************
********* OUTERLOOP: |X \setminus K| has exactly 2 taxa. ********
*****************************************************************
Found 4 trees before symmetry reduction:
(1,(2,(3,(X1,X2))));
(X1,(1,(2,(3,X2))));
(X2,(1,(2,(3,X1))));
((X1,X2),(1,(2,3)));
Found these unique patterns:
((Y,Y),(1,(2,3)));
(1,(2,(3,(Y,Y))));
(Y,(1,(2,(3,Y))));
3 patterns in total.
Selecting one tree per pattern.
(1,(2,(3,(X1,X2))));   (1,(2,(3,(Y,Y))));
(X1,(1,(2,(3,X2))));   (Y,(1,(2,(3,Y))));
((X1,X2),(1,(2,3)));   ((Y,Y),(1,(2,3)));
3 canonical trees generated.
Putting these canonical trees (without their pattern) in file ./firstTreesPreCherry.
(1,(2,(3,(X1,X2))));
(X1,(1,(2,(3,X2))));
((X1,X2),(1,(2,3)));
=======================
Expansions first tree:
(1,(2,(3,(X1,X2))));
(X1,(1,(2,(3,X2))));
((X1,X2),(1,(2,3)));
Removing expansions of first tree that have a cherry in the first tree
(outside the preserved region).
2 of the trees had such a cherry.
Expansions of first tree after removing trees with cherries:
(X1,(1,(2,(3,X2))));
=======================
Expansions second tree:
(1,(2,(3,(X1,X2))));
(X1,(1,(2,(3,X2))));
(X2,(1,(2,(3,X1))));
((X1,X2),(1,(2,3)));
=======================
Summary:
1 first trees
4 second trees
Total number of tree pairs:
4
=======================
======================
Combination 1 1 (i.e. tree pair 1)
(X1,(1,(2,(3,X2))));
(1,(2,(3,(X1,X2))));
Considering all the forests. (Singleton forests allowed = no)
FOREST1: x1 <= 1
OPT without modification: 1
OPT with modification: 1
FOREST2: x6 = 1
OPT without modification: 2
OPT with modification: 2
======================
Combination 1 2 (i.e. tree pair 2)
(X1,(1,(2,(3,X2))));
(X1,(1,(2,(3,X2))));
Considering all the forests. (Singleton forests allowed = no)
FOREST1: x1 <= 1
OPT without modification: 0
OPT with modification: 0
======================
Combination 1 3 (i.e. tree pair 3)
(X1,(1,(2,(3,X2))));
(X2,(1,(2,(3,X1))));
Considering all the forests. (Singleton forests allowed = no)
FOREST1: x1 <= 1
OPT without modification: 2
OPT with modification: 2
======================
Combination 1 4 (i.e. tree pair 4)
(X1,(1,(2,(3,X2))));
((X1,X2),(1,(2,3)));
Considering all the forests. (Singleton forests allowed = no)
FOREST1: x1 <= 1
OPT without modification: 1
OPT with modification: 1
FOREST2: x3 = 1
OPT without modification: 2
OPT with modification: 2
----------------------------------------------
----------------------------------------------
No counter-example found: BASE CASE TRUE!
----------------------------------------------
----------------------------------------------

\end{verbatim}

\subsubsection{$K$ is a 2-chain (and pendant in $T'$)}

Here the input file (2chain.input) contains the trees:
\begin{verbatim}
(a,(1,(2,b)));
(a,(1,2));
\end{verbatim}
Note how the second tree has a placeholder taxon missing: this indicates that
the cherry is pendant in the second tree.
The executation trace is here.

\begin{verbatim}
Reading trees from 2chainPendant.input.
Will add up to 2 taxa.
REMOVECHERRIES = yes
LOGS = no
STARTTAXA = 0
FORALL = no
ALLOWFORESTS = yes
ALLOWSINGLETONFORESTS = no
Note: in the execution below, FOREST1 indicates the situation F' = T'.
*****************************************************************
First tree: (a,(1,(2,b)));
Second tree: (a,(1,2));
*****************************************************************
Going to try expansions with 0 taxa, ... all the way up to 2 taxa.
*****************************************************************
********* OUTERLOOP: |X \setminus K| has exactly 0 taxa. ********
*****************************************************************
Found 1 trees before symmetry reduction:
(1,2);
Found these unique patterns:
(1,2);
1 patterns in total.
Selecting one tree per pattern.
(1,2);   (1,2);
1 canonical trees generated.
Putting these canonical trees (without their pattern) in file ./firstTreesPreCherry.
(1,2);
=======================
Expansions first tree:
(1,2);
Removing expansions of first tree that have a cherry in the first tree
(outside the preserved region).
0 of the trees had such a cherry.
Expansions of first tree after removing trees with cherries:
(1,2);
=======================
Expansions second tree:
(1,2);
=======================
Summary:
1 first trees
1 second trees
Total number of tree pairs:
1
=======================
======================
Combination 1 1 (i.e. tree pair 1)
(1,2);
(1,2);
Considering all the forests. (Singleton forests allowed = no)
FOREST1: x1 <= 1
OPT without modification: 0
OPT with modification: 0
*****************************************************************
********* OUTERLOOP: |X \setminus K| has exactly 1 taxa. ********
*****************************************************************
Found 2 trees before symmetry reduction:
(1,(2,X1));
(X1,(1,2));
Found these unique patterns:
(1,(2,Y));
(Y,(1,2));
2 patterns in total.
Selecting one tree per pattern.
(1,(2,X1));   (1,(2,Y));
(X1,(1,2));   (Y,(1,2));
2 canonical trees generated.
Putting these canonical trees (without their pattern) in file ./firstTreesPreCherry.
(1,(2,X1));
(X1,(1,2));
=======================
Expansions first tree:
(1,(2,X1));
(X1,(1,2));
Removing expansions of first tree that have a cherry in the first tree
(outside the preserved region).
0 of the trees had such a cherry.
Expansions of first tree after removing trees with cherries:
(1,(2,X1));
(X1,(1,2));
=======================
Expansions second tree:
(X1,(1,2));
=======================
Summary:
2 first trees
1 second trees
Total number of tree pairs:
2
=======================
======================
Combination 1 1 (i.e. tree pair 1)
(1,(2,X1));
(X1,(1,2));
Considering all the forests. (Singleton forests allowed = no)
FOREST1: x1 <= 1
OPT without modification: 0
OPT with modification: 0
======================
Combination 2 1 (i.e. tree pair 2)
(X1,(1,2));
(X1,(1,2));
Considering all the forests. (Singleton forests allowed = no)
FOREST1: x1 <= 1
OPT without modification: 0
OPT with modification: 0
*****************************************************************
********* OUTERLOOP: |X \setminus K| has exactly 2 taxa. ********
*****************************************************************
Found 4 trees before symmetry reduction:
(1,(2,(X1,X2)));
(X1,(1,(2,X2)));
(X2,(1,(2,X1)));
((X1,X2),(1,2));
Found these unique patterns:
((Y,Y),(1,2));
(1,(2,(Y,Y)));
(Y,(1,(2,Y)));
3 patterns in total.
Selecting one tree per pattern.
(1,(2,(X1,X2)));   (1,(2,(Y,Y)));
(X1,(1,(2,X2)));   (Y,(1,(2,Y)));
((X1,X2),(1,2));   ((Y,Y),(1,2));
3 canonical trees generated.
Putting these canonical trees (without their pattern) in file ./firstTreesPreCherry.
(1,(2,(X1,X2)));
(X1,(1,(2,X2)));
((X1,X2),(1,2));
=======================
Expansions first tree:
(1,(2,(X1,X2)));
(X1,(1,(2,X2)));
((X1,X2),(1,2));
Removing expansions of first tree that have a cherry in the first tree
(outside the preserved region).
2 of the trees had such a cherry.
Expansions of first tree after removing trees with cherries:
(X1,(1,(2,X2)));
=======================
Expansions second tree:
((X1,X2),(1,2));
=======================
Summary:
1 first trees
1 second trees
Total number of tree pairs:
1
=======================
======================
Combination 1 1 (i.e. tree pair 1)
(X1,(1,(2,X2)));
((X1,X2),(1,2));
Considering all the forests. (Singleton forests allowed = no)
FOREST1: x1 <= 1
OPT without modification: 1
OPT with modification: 1
FOREST2: x3 = 1
OPT without modification: 2
OPT with modification: 2
----------------------------------------------
----------------------------------------------
No counter-example found: BASE CASE TRUE!
----------------------------------------------
----------------------------------------------
\end{verbatim}

\subsection{Computational proof of base-case for interrupted-4-chain.}
\label{subsec:bcinterrupt}

Here $|X \setminus K| \leq 3$. The input trees (in the file int4chain.input) are:

\begin{verbatim}
(c,(1,(2,(e,(3,(4,d))))));
(a,(1,(2,(3,(4,b)))));
\end{verbatim}
The command is
\begin{verbatim}
./proofExistsPrune.sh int4chain.input 3
\end{verbatim}
The execution trace is too large to show here, but it is available on the GitHub page. It takes 2 seconds to complete, looping through 31 $(T,T')$ pairs in total (and in each case all forests $F'$ obtained from $T'$ that do not contain singletons: in total 61 $(T,F')$ pairs are considered). When $|X \setminus K|=3$ there is 1 candidate for $T$ and 12 candidates for $T'$. If we had not applied base case pruning, there would be 10 candidates for $T$. 

\subsection{Computational proof of base-case for preserved 2-2-2 star theorem.}
\label{subsec:bc32}

The input file (222star.input) consists of these two trees:

\begin{verbatim}
(a,(1,(2,((3,(4,b)),(5,(6,c))))));
(a,(1,(2,((3,(4,b)),(5,(6,c))))));
\end{verbatim}
The command is now
\begin{verbatim}
./proofExistsPrune.sh 222star.input 3
\end{verbatim}
This takes approximately 8 seconds to complete. In total 124 $(T,F')$ pairs are considered. See GitHub for the full execution trace.

\subsection{Computational proof of base-case for preserved 2-2-1 star theorem.}
\label{subsec:bc221}

The input file (221star.input) consists of these two trees:

\begin{verbatim}
(a,(1,(2,((3,(4,b)),(5,c)))));
(a,(1,(2,((3,(4,b)),(5,c)))));
\end{verbatim}
The command is now
\begin{verbatim}
./proofExistsPrune.sh 221star.input 3
\end{verbatim}
This takes approximately 9 seconds to complete. In total 124 $(T,F')$ pairs are considered. See GitHub for the full execution trace.

\subsection{Computational proof of base-case that common 4-chains are preserved in all MAFs}
\label{subsec:bc4all}

Here we use our code in ``for all'' mode. The input (4chain.input) is
\begin{verbatim}
(a,(1,(2,(3,(4,b)))));
(a,(1,(2,(3,(4,b)))));
\end{verbatim}
The command is now
\begin{verbatim}
./proofForallPrune.sh 4chain.input 2
\end{verbatim}
To highlight how ``for all'' mode is different we here provide the complete execution trace.

\begin{verbatim}
Reading trees from 4chain.input.
Will add up to 2 taxa.
REMOVECHERRIES = yes
LOGS = no
STARTTAXA = 0
FORALL = yes
ALLOWFORESTS = yes
ALLOWSINGLETONFORESTS = no
Note: in the execution below, FOREST1 indicates the situation F' = T'.
*****************************************************************
First tree: (a,(1,(2,(3,(4,b)))));
Second tree: (a,(1,(2,(3,(4,b)))));
*****************************************************************
Going to try expansions with 0 taxa, ... all the way up to 2 taxa.
*****************************************************************
********* OUTERLOOP: |X \setminus K| has exactly 0 taxa. ********
*****************************************************************
Found 1 trees before symmetry reduction:
(1,(2,(3,4)));
Found these unique patterns:
(1,(2,(3,4)));
1 patterns in total.
Selecting one tree per pattern.
(1,(2,(3,4)));   (1,(2,(3,4)));
1 canonical trees generated.
Putting these canonical trees (without their pattern) in file ./firstTreesPreCherry.
(1,(2,(3,4)));
=======================
Expansions first tree:
(1,(2,(3,4)));
Removing expansions of first tree that have a cherry in the first tree
(outside the preserved region).
0 of the trees had such a cherry.
Expansions of first tree after removing trees with cherries:
(1,(2,(3,4)));
=======================
Expansions second tree:
(1,(2,(3,4)));
=======================
Summary:
1 first trees
1 second trees
Total number of tree pairs:
1
=======================
======================
Combination 1 1 (i.e. tree pair 1)
(1,(2,(3,4)));
(1,(2,(3,4)));
Considering all the forests. (Singleton forests allowed = no)
FOREST1: x1 <= 1
OPT without modification: 0
OPT with modification: 1
1c1
< 0
---
> 1
*****************************************************************
********* OUTERLOOP: |X \setminus K| has exactly 1 taxa. ********
*****************************************************************
Found 2 trees before symmetry reduction:
(1,(2,(3,(4,X1))));
(X1,(1,(2,(3,4))));
Found these unique patterns:
(1,(2,(3,(4,Y))));
(Y,(1,(2,(3,4))));
2 patterns in total.
Selecting one tree per pattern.
(1,(2,(3,(4,X1))));   (1,(2,(3,(4,Y))));
(X1,(1,(2,(3,4))));   (Y,(1,(2,(3,4))));
2 canonical trees generated.
Putting these canonical trees (without their pattern) in file ./firstTreesPreCherry.
(1,(2,(3,(4,X1))));
(X1,(1,(2,(3,4))));
=======================
Expansions first tree:
(1,(2,(3,(4,X1))));
(X1,(1,(2,(3,4))));
Removing expansions of first tree that have a cherry in the first tree
(outside the preserved region).
0 of the trees had such a cherry.
Expansions of first tree after removing trees with cherries:
(1,(2,(3,(4,X1))));
(X1,(1,(2,(3,4))));
=======================
Expansions second tree:
(1,(2,(3,(4,X1))));
(X1,(1,(2,(3,4))));
=======================
Summary:
2 first trees
2 second trees
Total number of tree pairs:
4
=======================
======================
Combination 1 1 (i.e. tree pair 1)
(1,(2,(3,(4,X1))));
(1,(2,(3,(4,X1))));
Considering all the forests. (Singleton forests allowed = no)
FOREST1: x1 <= 1
OPT without modification: 0
OPT with modification: 1
1c1
< 0
---
> 1
======================
Combination 1 2 (i.e. tree pair 2)
(1,(2,(3,(4,X1))));
(X1,(1,(2,(3,4))));
Considering all the forests. (Singleton forests allowed = no)
FOREST1: x1 <= 1
OPT without modification: 1
OPT with modification: 2
1c1
< 1
---
> 2
======================
Combination 2 1 (i.e. tree pair 3)
(X1,(1,(2,(3,4))));
(1,(2,(3,(4,X1))));
Considering all the forests. (Singleton forests allowed = no)
FOREST1: x1 <= 1
OPT without modification: 1
OPT with modification: 2
1c1
< 1
---
> 2
======================
Combination 2 2 (i.e. tree pair 4)
(X1,(1,(2,(3,4))));
(X1,(1,(2,(3,4))));
Considering all the forests. (Singleton forests allowed = no)
FOREST1: x1 <= 1
OPT without modification: 0
OPT with modification: 1
1c1
< 0
---
> 1
*****************************************************************
********* OUTERLOOP: |X \setminus K| has exactly 2 taxa. ********
*****************************************************************
Found 4 trees before symmetry reduction:
(1,(2,(3,(4,(X1,X2)))));
(X1,(1,(2,(3,(4,X2)))));
(X2,(1,(2,(3,(4,X1)))));
((X1,X2),(1,(2,(3,4))));
Found these unique patterns:
((Y,Y),(1,(2,(3,4))));
(1,(2,(3,(4,(Y,Y)))));
(Y,(1,(2,(3,(4,Y)))));
3 patterns in total.
Selecting one tree per pattern.
(1,(2,(3,(4,(X1,X2)))));   (1,(2,(3,(4,(Y,Y)))));
(X1,(1,(2,(3,(4,X2)))));   (Y,(1,(2,(3,(4,Y)))));
((X1,X2),(1,(2,(3,4))));   ((Y,Y),(1,(2,(3,4))));
3 canonical trees generated.
Putting these canonical trees (without their pattern) in file ./firstTreesPreCherry.
(1,(2,(3,(4,(X1,X2)))));
(X1,(1,(2,(3,(4,X2)))));
((X1,X2),(1,(2,(3,4))));
=======================
Expansions first tree:
(1,(2,(3,(4,(X1,X2)))));
(X1,(1,(2,(3,(4,X2)))));
((X1,X2),(1,(2,(3,4))));
Removing expansions of first tree that have a cherry in the first tree
(outside the preserved region).
2 of the trees had such a cherry.
Expansions of first tree after removing trees with cherries:
(X1,(1,(2,(3,(4,X2)))));
=======================
Expansions second tree:
(1,(2,(3,(4,(X1,X2)))));
(X1,(1,(2,(3,(4,X2)))));
(X2,(1,(2,(3,(4,X1)))));
((X1,X2),(1,(2,(3,4))));
=======================
Summary:
1 first trees
4 second trees
Total number of tree pairs:
4
=======================
======================
Combination 1 1 (i.e. tree pair 1)
(X1,(1,(2,(3,(4,X2)))));
(1,(2,(3,(4,(X1,X2)))));
Considering all the forests. (Singleton forests allowed = no)
FOREST1: x1 <= 1
OPT without modification: 1
OPT with modification: 2
1c1
< 1
---
> 2
FOREST2: x7 = 1
OPT without modification: 2
OPT with modification: 3
1c1
< 2
---
> 3
======================
Combination 1 2 (i.e. tree pair 2)
(X1,(1,(2,(3,(4,X2)))));
(X1,(1,(2,(3,(4,X2)))));
Considering all the forests. (Singleton forests allowed = no)
FOREST1: x1 <= 1
OPT without modification: 0
OPT with modification: 1
1c1
< 0
---
> 1
======================
Combination 1 3 (i.e. tree pair 3)
(X1,(1,(2,(3,(4,X2)))));
(X2,(1,(2,(3,(4,X1)))));
Considering all the forests. (Singleton forests allowed = no)
FOREST1: x1 <= 1
OPT without modification: 2
OPT with modification: 3
1c1
< 2
---
> 3
======================
Combination 1 4 (i.e. tree pair 4)
(X1,(1,(2,(3,(4,X2)))));
((X1,X2),(1,(2,(3,4))));
Considering all the forests. (Singleton forests allowed = no)
FOREST1: x1 <= 1
OPT without modification: 1
OPT with modification: 2
1c1
< 1
---
> 2
FOREST2: x3 = 1
OPT without modification: 2
OPT with modification: 3
1c1
< 2
---
> 3
----------------------------------------------
----------------------------------------------
No counter-example found: BASE CASE TRUE!
----------------------------------------------
----------------------------------------------
\end{verbatim}

\subsection{Blocking gadgets}

\begin{lem}
\label{lem:block}
Let $K = \{a,b,c,d\}$ and $L=\{e,f,g,h\}$ be length-4 common chains which are arranged next to each other in the order $a,b,c,d,e,f,g,h$ in $T$ and $d,c,b,a,h,g,f,e$ in $T'$. Then \emph{every} MAF of $T,T'$ contains a component $K$ (with no other taxa in it) and a component $L$ (with no other taxa in it).
\end{lem}
\begin{proof}
Let $F^{*}$ be an arbitrary MAF for $T, T$'. From Lemma \ref{lem:4everywhere}, $K$ is preserved by $F^{*}$ and so is $L$. They cannot be in the same component of $F^{*}$, because $T|(K \cup L)$ has a different topology to $T'|(K \cup L)$: this is the reason why in $T'$ the relative orientations of $K$ and $L$ are opposite to $T$. Now, suppose the component that contains $K$, contains some taxon $x \not \in K \cup L$. In $T$, $x$ must meet $K$ at its $a$ end, but that is not possible because in $T'$ the $a$ end of $K$ is blocked by the presence of $L$. A symmetrical argument holds for the $L$ component.
\end{proof}

\end{document}